\documentclass{IEEEtran}

\usepackage{cite}
\usepackage{amsmath,amssymb,amsfonts,amsthm}
\usepackage{graphicx}
\usepackage{hyperref}
\hypersetup{hidelinks}
\usepackage{textcomp}

\usepackage{enumerate}
\usepackage{xcolor}
\usepackage{flushend}

\graphicspath{{Figs/}}

\definecolor{myred}{rgb}{0.7,0.1,0.16}
\definecolor{myblue}{rgb}{0,0.32,0.7}
\definecolor{mygreen}{rgb}{0.133,0.545,0.133}

\usepackage{tikz}
\usetikzlibrary{arrows}

\newtheorem{theorem}{Theorem}
\newtheorem{remark}{Remark}

\newtheorem{definition}{Definition}
\newtheorem{corollary}{Corollary}

\def\BibTeX{{\rm B\kern-.05em{\sc i\kern-.025em b}\kern-.08em
    T\kern-.1667em\lower.7ex\hbox{E}\kern-.125emX}}

\begin{document}
\title{Universal Formula Families for Safe Stabilization of Single-Input Nonlinear Systems}
\author{Bo Wang and Miroslav Krsti{\'c}
\thanks{This work was partially supported by the PSC-CUNY Research Award from The City University of New York. \textit{(Corresponding author: Bo Wang.)}}
\thanks{Bo Wang is with the Department of Mechanical Engineering, The City College of New York, The City University of New York, New York, NY 10031, USA (e-mail: bwang1@ccny.cuny.edu). }
\thanks{Miroslav Krsti{\'c} is with the Department of Mechanical and Aerospace Engineering, University of California San Diego, La Jolla, CA 92093, USA (e-mail: mkrstic@ucsd.edu). }
}

\maketitle
\begin{abstract}
We develop an optimization-free framework for safe stabilization of single-input control-affine nonlinear systems with a given control Lyapunov function (CLF) and a given control barrier function (CBF), where the desired equilibrium lies in the interior of the safe set. An explicit compatibility condition is derived that is necessary and sufficient for the pointwise simultaneous satisfaction of the CLF and CBF inequalities. When this condition holds, two closed-form continuous state-feedback laws are constructed from the Lie-derivative data of the CLF and CBF via standard universal stabilizer formulas, yielding asymptotic stabilization of the origin and forward invariance of the interior of the safe set, without online quadratic programming. The two laws belong to broader families parametrized by a free nondecreasing function, providing additional design flexibility. When the compatibility condition fails, a safety-prioritizing modification preserves forward invariance and drives the state toward the safe-set boundary until a compatible region is reached, whereupon continuity at the origin and asymptotic stabilization are recovered. The framework produces families of explicit constructive 
alternatives to CLF-CBF quadratic programming for scalar-input 
nonlinear systems.

\end{abstract}

\section{Introduction} \label{sec:introduction}

Control barrier function (CBF)-based techniques have proven effective for enforcing safety constraints \cite{ames2014control,ames2017control}, with applications spanning walking robots \cite{ames2019control}, automotive systems \cite{xu2018correctness,han2024safety,wang2026further}, and multi-agent systems \cite{jankovic2024multiagent}, among others.

\paragraph*{Safe Stabilization as Compatibility, Not Optimization}
Quadratic programming (QP) has become the dominant mechanism for combining control Lyapunov functions (CLFs) and CBFs, yielding CLF-CBF-based QP \cite{ames2014control,xu2015robustness,ames2017control,ames2019control}, CBF-only safety filters \cite{xu2015robustness,gurriet2018towards,ames2019control,singletary2021safety,krstic2023inverse}, and $\gamma m$-CLF-CBF QP approaches \cite{jankovic2018robust,wang2026further}---all extensions of Freeman and Kokotovi\'c's pointwise minimum-norm (PMN) controller \cite{freeman1996robust}. However, these optimization-based approaches can degrade or lose stabilization, introduce undesirable equilibria, and compromise boundedness \cite{saberi2002constrained,grover2020does,reis2021control,jankovic2018robust}, particularly when slack variables are introduced. Yet the central issue is not optimization-theoretic: at each state, the CLF and CBF conditions impose affine inequalities in the control input, which in the single-input case reduce to just two scalar inequalities.

\paragraph*{Scalar Structure: Two Inequalities, One Question}
Safe stabilization is therefore fundamentally a compatibility question: does there exist a continuous state-feedback law that simultaneously satisfies both inequalities at every state?
This is conceptually distinct from the feasibility of a particular QP formulation. Optimization-based controllers provide one mechanism for enforcing inequalities, but they do not isolate the structural conditions under which simultaneous safety and asymptotic stabilization are achievable---so issues such as loss of stabilization under safety filtering tend to be analyzed within specific QP architectures rather than at the level of existence of safe stabilizing feedback laws. The present work addresses this foundational question directly.

\paragraph*{Universal Formulas} Artstein's theorem establishes that CLF existence is equivalent to stabilizability via continuous state feedback, but its proof is nonconstructive \cite{artstein1983stabilization}. This motivated Sontag \cite{sontag1989universal} to derive an explicit formula from Lie-derivative data, later extended to bounded and positive controls \cite{lin1991universal,lin1995control}. Freeman and Kokotovi\'c's PMN controller \cite{freeman1996robust} additionally yields favorable stability margins and inverse optimality \cite{freeman1996inverse}---properties extended to safety filters in \cite{krstic2023inverse}. These universal formulas underpin a broader program of constructive stabilization \cite{wu2005simultaneous}.

Ong and Cortés \cite{ong2019universal} extended this program to safe stabilization via weighted centroids of the admissible control set. While their construction applies to general admissible control sets, the resulting feedback requires two integrals over that set, obscuring analytical structure. Li and Sun \cite{li2023graphical} proposed an alternative via PMN-QP. The present paper instead derives families of closed-form 
safe-stabilizing feedback laws directly from the CLF-CBF 
inequalities, parametrized by a free nondecreasing function. Figure~\ref{fig:framework} illustrates the simplest realization, under the necessary and sufficient condition, of the proposed blending architecture, in which a CLF universal stabilizer and a CBF universal safe formula are coordinated to simultaneously achieve stabilization and safety.

\paragraph*{Relation to CLF-CBF and Control-Sharing Results}
Xu \cite{xu2018constrained} characterized control sharing for multiple CBFs in single-input systems via necessary and sufficient conditions for a common input; Cohen et al.\ \cite{cohen2025compatibility} extended analogous compatibility conditions to multi-input systems. Both works concern compatibility among barrier constraints. Mestres et al.\ \cite{mestres2025neural} addressed multiple inequality constraints through a strictly convex minimizer approximated by a neural network. The present paper instead targets the scalar CLF-CBF setting and delivers explicit compatibility conditions together with closed-form continuous feedback laws, without neural-network approximation or online optimization.

\paragraph*{Contributions --- Explicit Compatibility and Closed-Form Safe Stabilization}
This paper considers the safe-stabilization problem for single-input nonlinear systems with a given CLF and CBF, where the desired equilibrium lies in the interior of the safe set. We adopt a strict barrier definition, precluding the state from residing on the boundary. For the scalar-input case, we derive an explicit compatibility condition for the pointwise simultaneous satisfaction of the CLF and CBF inequalities and use it to construct continuous safe-stabilizing feedback laws in closed form. When the compatibility condition holds, the proposed controllers guarantee asymptotic stabilization of the origin together with forward invariance of the interior of the safe set. When it fails, we develop safety-prioritizing continuous feedback laws that preserve forward invariance and recover continuity and asymptotic stabilization near the origin. The paper thereby provides an explicit, constructive, and optimization-free alternative to CLF-CBF-based QP methods for scalar-input nonlinear systems. In summary, the paper contributes the following: 

\begin{enumerate}
    \item[(i)] We identify an explicit compatibility condition that is necessary and sufficient for the pointwise simultaneous satisfaction of the CLF and CBF inequalities, precisely characterizing when safety and stabilization can be achieved simultaneously in the single-input case.
    \item[(ii)] We construct families of explicit continuous state-feedback 
laws in closed form, parametrized by a free nondecreasing 
function, applying standard universal stabilizer formulas 
directly to the Lie-derivative data of the CLF and CBF.
    \item[(iii)] When the compatibility condition fails, we develop a safety-prioritizing continuous feedback law that preserves forward invariance of the safe set, steers trajectories toward the safe-set boundary in incompatible regions, and recovers continuity at the origin and asymptotic stabilization under the standing assumptions.
\end{enumerate}

\begin{figure}
    \centering
    \includegraphics[width=0.8\linewidth]{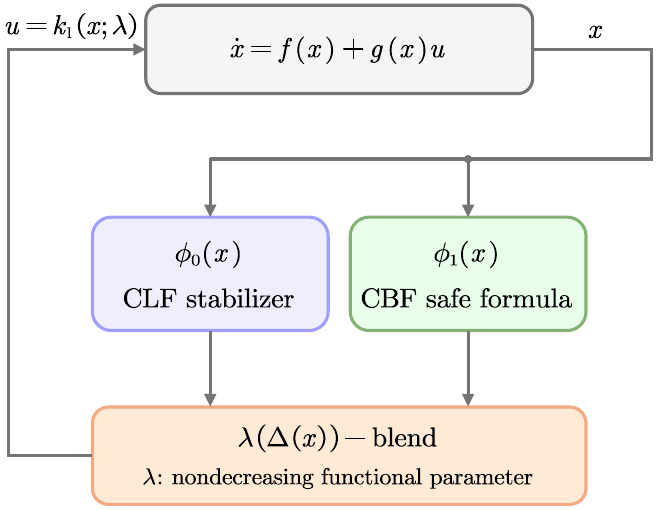}
    \caption{Simplest realization of the proposed architecture, under the necessary and sufficient condition, for simultaneously achieving stabilization and safety.}
    \label{fig:framework}
\end{figure}

\paragraph*{Comparison with Prior Work} Compared with the modified-QP designs \cite{jankovic2018robust}, we establish a sharper construction, for a compatible CLF-CBF pair, yielding families of feedbacks built directly from Lie-derivative data and standard universal formulas. Unlike \cite{ong2019universal}, which relies on integral expressions over the admissible control set, our feedback laws are explicit and closed-form. Unlike the control-sharing framework \cite{xu2018constrained}, which concerns compatibility among multiple CBF constraints within a QP-based tracking formulation, the present paper studies compatibility between a CLF inequality and a CBF inequality in the scalar-input setting and yields explicit continuous safe-stabilizing feedback laws.

Relative to \cite{li2023graphical}, which handles CLF-CBF incompatibility by relaxing the CLF inequality with a slack variable within a PMN-QP framework, this paper develops safety-prioritizing closed-form feedbacks for incompatible regions, preserving forward invariance and recovering continuity and asymptotic stabilization without slack variables or relaxation weights. Relative to \cite{mestres2025neural}, which defines the controller as the unique minimizer of a strictly convex objective approximated by a neural network, the distinction is both conceptual and computational. Both approaches operate over the same feasible interval induced by the CLF and CBF inequalities, but the controller in \cite{mestres2025neural} requires solving a quartic equation or neural-network approximation whose guarantees are restricted to compact subsets of the parameter domain, limiting stabilization to semiglobal results. Our approach exploits the one-dimensional geometry to obtain explicit closed-form feedback laws that enforce feasibility under a clear compatibility condition.

More broadly, this paper is a conceptual extension of constrained  stabilization  arising in nonovershooting control \cite{krstic2006nonovershooting}; see also \cite[Sect.~XI]{krstic2023inverse}. In that setting, safety and stabilization were unified under state constraints with the equilibrium  on the boundary of the admissible set. The present work revisits this agenda for general single-input control-affine nonlinear systems with a CLF and CBF, studying when safety and asymptotic stabilization can be achieved simultaneously. 
The treatment of disturbances, via ISS-CLFs or robust CBF formulations as in \cite{jankovic2018robust}, is beyond the scope of this note.

\paragraph*{Organization} Section \ref{sec:preliminaries} presents preliminaries and the problem statement. Section \ref{sec:main} develops the main results on compatibility conditions and universal formulas for safe stabilization of single-input nonlinear systems. Section \ref{sec:simulation} provides simulation examples. Section \ref{sec:conclusion} provides concluding remarks.

\section{Preliminaries and Problem Statement}\label{sec:preliminaries}

\textit{Notation}: Let $|\cdot|$ denote the Euclidean norm on $\mathbb{R}^n$. For $S\subset\mathbb{R}^n$, $\partial S$ and $\operatorname{Int}(S)$ denote its boundary and interior. $\mathcal{K}$ is the class of continuous, strictly increasing functions $\alpha:\mathbb{R}_{\ge0}\to\mathbb{R}_{\ge0}$ with $\alpha(0)=0$; $\mathcal{K}_\infty\subset\mathcal{K}$ are those unbounded. A function $\alpha_h:\mathbb{R}\to\mathbb{R}$ is in $\mathcal{K}_\infty^e$ if it is strictly increasing, $\alpha_h(0)=0$, and $\alpha_h(s)\to\pm\infty$ as $s\to\pm\infty$. Arguments are omitted when clear.

Consider the single-input control-affine nonlinear system
\begin{equation}\label{eq:nls}
    \dot{x}=f(x)+g(x)u,
\end{equation}
where $x\in\mathbb{R}^n$ is the state, $u\in\mathbb{R}$ is the control input, and $f:\mathbb{R}^n\to \mathbb{R}^n$ and $g:\mathbb{R}^n\to \mathbb{R}^n$ are known smooth vector fields. We assume $f(0)=0$ so that the origin is an equilibrium of the unforced system. The following definitions are standard.

\begin{definition}[CLF]\rm A positive definite, proper, continuously
differentiable function $V:\mathcal{D}\to \mathbb{R}_{\ge 0}$ is a (local) \textit{CLF} for system \eqref{eq:nls} on $\mathcal{D}\subset\mathbb{R}^n$ if there exists a function $\alpha\in\mathcal{K}$ such that for all $x\in\mathcal{D}\backslash \{0\}$
\begin{equation}\label{eq:defCLF}
    L_gV(x)=0 \implies L_fV(x) + \alpha(|x|)<0.
\end{equation}
If $\mathcal{D}=\mathbb{R}^n$, then $V$ is a global CLF for \eqref{eq:nls}.
\end{definition}

\begin{definition}[SCP]\rm A CLF $V$ for system \eqref{eq:nls} is said to have the \textit{small control property (SCP)} if for each $\varepsilon>0$ there is a $\delta>0$ such that, if $x\ne 0$ satisfies $|x|\le \delta$, then there exists $u$ with $|u|\le \varepsilon$ such that 
\begin{equation}
    L_fV(x)+L_gV(x) u <0.
\end{equation}
\end{definition}

Suppose there exists a state feedback controller for \eqref{eq:nls} such that the origin of the closed-loop system is asymptotically stable, then by the converse Lyapunov theorem, there exists a CLF \cite{KHALIL2002}. Furthermore, if there exists a stabilizer continuous at the origin, then a CLF exists with SCP \cite{sontag1989universal}.

Let $V$ be a (global) CLF for system \eqref{eq:nls} and let us denote
\begin{equation*}
    a_0(x):=L_fV(x)+\alpha(|x|),\quad b_0(x):=L_gV(x).
\end{equation*}
Various universal formulas for stabilization have been proposed in the literature, e.g., \cite{sontag1989universal,freeman1996robust}. It is shown that if $V$ satisfies the SCP, then the functions $\phi_{\rm S}$ and $\phi_{\rm F}$ below are continuous and give global
asymptotic stabilizers:
\begin{equation}
    \phi_{\rm S}(a_0(x),b_0(x)):=\left\{ \begin{array}{ll}
        -\frac{a_0(x)+\sqrt{a_0^2(x)+b_0^4(x)}}{b_0(x)} & \text{if}~b_0(x)\ne 0, \\
         0& \text{if}~b_0(x)= 0,
    \end{array}  \right.
\end{equation}
\begin{equation}
    \phi_{\rm F}(a_0(x),b_0(x)):=\left\{ \begin{array}{ll}
        -\frac{\max\{a_0(x)+b_0^2(x),0\}}{b_0(x)}, & \text{if}~b_0(x)\ne 0, \\
         0& \text{if}~b_0(x)= 0.
    \end{array}  \right.
\end{equation}
Note that $\phi_{\rm S}(a_0(x),b_0(x))$ is also smooth on $\mathbb{R}^n\backslash \{0\}$.

The property of safety is formalized by requiring the trajectories of the closed-loop system to remain within a safe set at all times. We assume that the safe set $\mathcal{C}$ is built as the 0-superlevel set of a smooth function $h:\mathbb{R}^n\to \mathbb{R}$, i.e., 
\begin{equation*}
    \mathcal{C}:=\{x\in\mathbb{R}^n:h(x)\ge 0\}.
\end{equation*}

\vspace{0cm}
\begin{definition}[CBF] \rm
A continuously differentiable function $h:\mathbb{R}^n\to \mathbb{R}$ is a \textit{CBF} with respect to the set $\mathcal{C}\subset\mathbb{R}^n$ if there exists $\alpha_h\in\mathcal{K}^e_\infty$ such that for all $x\in\mathbb{R}^n$
\begin{equation}\label{eq:defCBF}
    L_gh(x)=0 \implies L_fh(x) + \alpha_h(h(x))>0.
\end{equation}
\end{definition}

Suppose that the safe set $\mathcal{C}$ is compact, and there exists a state feedback controller such that $\mathcal{C}$ is safe for \eqref{eq:nls}, then by the converse theorem \cite[Theorem 3]{ames2019control}, $h|_\mathcal{C}:\mathcal{C}\to\mathbb{R}$ is a CBF. More converse results on CBFs can be found in \cite{mestres2024converse}.

Let $h$ be a CBF for \eqref{eq:nls} with respect to $\mathcal{C}$ and let us denote
\begin{equation*}
    a_1(x):=-L_fh(x)-\alpha_h(h(x)),\quad b_1(x):=-L_gh(x).
\end{equation*}
Note that the existence of a global CLF $V$ for \eqref{eq:nls} implies that for each $x\in\mathbb{R}^n\backslash\{0\}$ there exists $u\in\mathbb{R}$ enforcing the inequality
\begin{equation}\label{eq:clf}
    F_0:= a_0(x)+b_0(x)u<0.
\end{equation}
Similarly, the existence of a CBF $h$ implies that for each $x\in\mathbb{R}^n$ there exists $u\in\mathbb{R}$ that enforces the inequality
\begin{equation}\label{eq:cbf}
    F_1:=a_1(x)+b_1(x)u<0.
\end{equation}
However, the existence of CLF $V$ and CBF $h$ for \eqref{eq:nls} does not guarantee the existence of $u\in\mathbb{R}$ simultaneously enforcing \eqref{eq:clf} and \eqref{eq:cbf}.

\textit{Problem Statement:} Seek, first, a necessary condition for the existence of a control input $u\in\mathbb{R}$ satisfying \eqref{eq:clf} and \eqref{eq:cbf}. Then, construct a continuous feedback $k:\mathbb{R}^n\to\mathbb{R}$ such that, when the condition holds, 
\begin{equation}
u = k(x)
\end{equation}
ensures asymptotic stability of the origin and forward invariance of the interior of the safe set $\mathcal{C}$, and,  when the condition fails, safety of $\mathcal{C}$ is preserved.

\section{Main Results}\label{sec:main}

\subsection{A Necessary Condition for Safe Stabilization}

For each $x\in\mathbb{R}^n$, $b_0(x)$ and $b_1(x)$ are scalars. If $b_0(x)$ and $b_1(x)$ have the same sign, then it is always possible to find $u\in\mathbb{R}$ such that \eqref{eq:clf}-\eqref{eq:cbf} hold simultaneously. However, if $b_0(x)$ and $b_1(x)$ have opposite signs, then $u$ exerts opposite effects on $F_0$ and $F_1$, and \eqref{eq:clf}-\eqref{eq:cbf} may not be satisfied simultaneously.

The following theorem presents a necessary condition for the existence of safe stabilizing state-feedback controllers for the single-input nonlinear system \eqref{eq:nls}.

\begin{theorem}\label{thm:1}
    Consider the system~\eqref{eq:nls} and assume that there exist a CLF $V$ and a CBF $h$.    
    If there exists a function $k:\mathbb{R}^n \to \mathbb{R}$ such that, with $u = k(x)$, \eqref{eq:clf}-\eqref{eq:cbf} hold simultaneously for all $x\in\mathbb{R}^n\backslash\{0\}$, then for each $x$ such that $b_0(x)b_1(x)<0$, defining $i$ as the index with $b_i(x)>0$ and $j$ as the index with $b_j(x)<0$, it follows that
    \begin{equation}\label{eq:10}
        -\frac{a_j(x)}{b_j(x)}<-\frac{a_i(x)}{b_i(x)}.
    \end{equation}
\end{theorem} 
    
\begin{proof}
    Since $b_i(x)>0$ and $b_j(x)<0$, it follows from inequalities \eqref{eq:clf}-\eqref{eq:cbf} that
    \begin{equation}\label{eq:11}
        k(x)<-\frac{a_i(x)}{b_i(x)}\quad \text{and} \quad k(x)>-\frac{a_j(x)}{b_j(x)}.
    \end{equation}
   The existence of $k(x)$ satisfying \eqref{eq:11} implies \eqref{eq:10}, which completes the proof.   
\end{proof}

\subsection{Universal Formulas for Safe Stabilization}

Although Theorem \ref{thm:1} is simple in form, remarkably, in the single-input case, the necessary condition provided in Theorem \ref{thm:1} is also sufficient for the existence of a safe stabilizing control law 
$u=k(x)$ such that \eqref{eq:clf}-\eqref{eq:cbf} hold simultaneously.

Let us denote the Sontag-type and Freeman-type safety control laws, respectively, as
\begin{equation}
    \phi_{\rm S}(a_1(x),b_1(x)):=\left\{ \begin{array}{ll}
        -\frac{a_1(x)+\sqrt{a_1^2(x)+b_1^4(x)}}{b_1(x)} & \text{if}~b_1(x)\ne 0, \\
         0& \text{if}~b_1(x)= 0,
    \end{array}  \right.
\end{equation}
\begin{equation}
    \phi_{\rm F}(a_1(x),b_1(x)):=\left\{ \begin{array}{ll}
        -\frac{\max\{a_1(x)+b_1^2(x),0\}}{b_1(x)} & \text{if}~b_1(x)\ne 0, \\
         0& \text{if}~b_1(x)= 0.
    \end{array}  \right.
\end{equation}
The feedback $\phi_{\rm S}(a_1(x),b_1(x))$ is algebraically equivalent to the scalar-input specialization of the Sontag-type safe controller \cite[(13)--(14)]{krstic2023inverse}, which guarantees both safety and inverse optimality.
In the following, for clarity and brevity, let $(\phi_0,\phi_1)$ generically denote $(\phi_{\rm S}(a_0,b_0),\phi_{\rm S}(a_1,b_1))$ or $(\phi_{\rm F}(a_0,b_0),\phi_{\rm F}(a_1,b_1))$.
Denote $b(x):=[b_0(x)~b_1(x)]^\top$ and
\begin{equation}
    \Delta(x):=-\frac{a_0(x)}{b_0(x)}-\frac{a_1(x)}{b_1(x)}.
\end{equation}
Next, let $\lambda:\mathbb{R}\to(0,1)$ be a smooth nondecreasing function satisfying
\begin{equation}
\lim_{z\to -\infty}\lambda(z)=0,
\quad \text{and} \quad
\lim_{z\to +\infty}\lambda(z)=1.
\end{equation}
For example, one may choose $\lambda$ as the logistic function
\begin{equation}\label{eq:lambda}
    \lambda(z):=\frac{1}{1+e^{-z}},
\end{equation}
while other possible choices include
\begin{equation}
    \lambda(z):=\frac{1+\tanh(z)}{2}
    \quad \text{and} \quad
    \lambda(z):=\frac{z+\sqrt{1+z^2}}{2\sqrt{1+z^2}}.
\end{equation}
Hence, we do not merely propose a single universal formula, or a pair of universal formulas, but a family of universal formulas parametrized by the nondecreasing function $\lambda(\cdot)$. 

Let us define $k_{\rm l}:\mathbb{R}^n\to \mathbb{R}$ and $k_{\rm m}:\mathbb{R}^n\to \mathbb{R}$ as
    \begin{equation}
        k_{\rm l}(x):=k_{\rm m}(x):=\left \{
        \begin{array}{ll}
         \min\{\phi_0,\phi_1\}\quad& \text{if}~b_0, b_1\ge 0, |b|\ne 0, \\
         \max\{\phi_0,\phi_1\}& \text{if}~b_0, b_1\le 0, |b|\ne 0,\\
         0& \text{if}~b=0.
    \end{array}
        \right.
    \end{equation}
    If $b_0\cdot b_1<0$, define $i$ as the index with $b_i(x)>0$ and $j$ as the index with $b_j(x)<0$,
\begin{equation}
    k_{\rm l}(x):=\left(1-\lambda\left(\Delta\right) \right)\max\left\{\phi_i,-\frac{a_j}{b_j} \right\}+\lambda\left(\Delta \right)\min\left\{\phi_j,-\frac{a_i}{b_i} \right\} \label{eq:kl}
\end{equation}
and
\begin{equation}\label{eq:km}
    k_{\rm m}(x):=\min\left\{\phi_j,\max \left\{ \phi_i,\frac{\Delta}{2}\right\}\right\}.
\end{equation}

The following is one of the paper's main results.
\begin{theorem}\label{thm:2}
    Consider the system~\eqref{eq:nls} and assume that there exist a CLF $V$ and a CBF $h$. If, 
    for all $x\in\mathbb{R}^n$,
    \begin{equation}\label{eq:14}
        b_0(x)\cdot b_1(x)<0\implies-\frac{a_j(x)}{b_j(x)}<-\frac{a_i(x)}{b_i(x)},
    \end{equation}
    where the indices $i,j\in \{0,1\}$ denote the unique indices such that $b_i(x)>0$ and $b_j(x)<0$, 
    then there exists a function $k:\mathbb{R}^n \to \mathbb{R}$ such that, with $u = k(x)$, \eqref{eq:clf}-\eqref{eq:cbf} hold simultaneously for all $x\in\mathbb{R}^n\backslash \{0\}$. 
    Furthermore, if \eqref{eq:14} holds, then both the state-feedback laws $u = k_{\rm l}(x)$ and $u = k_{\rm m}(x)$ are continuous on $\mathbb{R}^n\backslash \{0\}$ and ensure that \eqref{eq:clf}-\eqref{eq:cbf} hold simultaneously.
\end{theorem} 

\begin{proof}
    See Appendix \ref{appexdix:A}.
\end{proof}

\begin{remark}\rm 
In \cite{xu2018constrained}, control sharing of multiple CBFs is characterized via necessary and sufficient conditions for a common input. The condition in Theorem~\ref{thm:2} for a controller $u=k(x)$ satisfying \eqref{eq:clf}–\eqref{eq:cbf} is analogous, here for a CLF–CBF pair. The contribution is an explicit continuous stabilizing feedback via universal formulas, yielding a closed-form controller that enforces safety and asymptotic stability without online optimization.
\end{remark}

\begin{remark}\rm 
The function $\max\{x,0\}$ is continuous but not differentiable at $x=0$. It admits smooth approximations, e.g., $\varepsilon\ln(1+e^{x/\varepsilon})\to \max\{x,0\}$ as $\varepsilon\to 0$. Using $\max\{x,y\}=y+\max\{x-y,0\}$, $\min\{x,y\}=y-\max\{y-x,0\}$, and $\min\{x,0\}=x-\max\{x,0\}$, the laws in \eqref{eq:kl}--\eqref{eq:km} can be approximated by smooth functions.
\end{remark}

\begin{remark}\rm 
    When $b_0(x) \cdot b_1(x) < 0$, an equivalent and ``symmetric’’ expression for $k_{\rm m}(\cdot)$ is the following: 
\begin{equation}
k_{\rm m}(x)
:= \max\left\{\phi_i, \min\left\{ \phi_j, \frac{\Delta}{2}\right\} \right\}.
\end{equation}
    In \eqref{eq:km}, the term $\frac{1}{2}\big(-\frac{a_i}{b_i}-\frac{a_j}{b_j}\big)$ may be replaced by any convex combination of  $-\frac{a_i}{b_i}$ and $-\frac{a_j}{b_j}$, i.e., 
    \begin{equation}\label{eq:kma}
    k_{{\rm m},\eta}(x):=\min\left\{\phi_j,\max \left\{ \phi_i,-(1-\eta)\frac{a_i}{b_i}-\eta\frac{a_j}{b_j}\right\}\right\},
\end{equation}
where $\eta\in(0,1)$; the proof and result are unchanged. This flexibility is used in the next section when the condition fails.
\end{remark}

In Theorem~\ref{thm:2}, two feedback laws, $k_{\rm l}$ and $k_{\rm m}$, are constructed for safe stabilization. In both cases, the pair $(\phi_0,\phi_1)$ can be instantiated using either Sontag’s universal formula $(\phi_{\rm S}(a_0,b_0),\phi_{\rm S}(a_1,b_1))$ or the PMN formula $(\phi_{\rm F}(a_0,b_0),\phi_{\rm F}(a_1,b_1))$. More generally, other universal formulas satisfying the standard regularity properties required in this context may be employed, including those proposed in \cite{lin1991universal,lin1995control}.

\subsection{Ensuring Safety with Incompatible CLFs and CBFs}

Theorem~\ref{thm:2} shows that if the CLF $V$ and CBF $h$ are compatible, i.e., \eqref{eq:14} holds, then safe stabilizers exist, e.g., $u=k_{\rm l}(x)$ or $u=k_{\rm m}(x)$. When \eqref{eq:14} fails, safety must be prioritized. One can modify $k_{\rm l}$ or $k_{\rm m}$ to preserve forward invariance while driving the state toward the safe-set boundary; once compatibility is restored, the controller resumes simultaneous safety and stabilization.

We define the following modified feedback control laws:
\begin{equation}
  k_{\rm l}^*(x) :=
  \begin{cases}
    \displaystyle 
    -\dfrac{a_1}{b_1},  &\text{if } b_0 b_1 < 0 \text{ and \eqref{eq:14} does not hold}, \\
    k_{\rm l}(x), 
      & \text{otherwise,}
  \end{cases}
\end{equation}
and
\begin{equation}
  k_{\rm m}^*(x) :=
  \begin{cases}
    \displaystyle 
    \min\!\left\{\phi_j,\; \max\!\left\{\phi_i,\; -\dfrac{a_1}{b_1}\right\}\right\}=-\dfrac{a_1}{b_1}, & \\
    
      &\hspace{-4cm} \text{if } b_0 b_1 < 0 \text{ and \eqref{eq:14} does not hold}, \\
    k_{\rm m}(x), 
      &\hspace{-4cm} \text{otherwise.}
  \end{cases}
\end{equation}

\begin{theorem}\label{thm:3}
    Consider the system~\eqref{eq:nls} and assume that there exist a CLF $V$ satisfying the SCP and a CBF $h$, and that $0 \in \operatorname{Int}(C)$. Then,
    \begin{enumerate}[(i)] 
        \item The control laws $u=k_{\rm l}^*(x)$ and $u=k_{\rm m}^*(x)$ are well defined on $\mathbb{R}^n$ and continuous on $\mathbb{R}^n\backslash\{0\}$.

        \item Under the control law $u=k_{\rm l}^*(x)$ or $u=k_{\rm m}^*(x)$, the set $\operatorname{Int}(C)$ is forward invariant.

        \item Under the control law $u=k_{\rm l}^*(x)$ or $u=k_{\rm m}^*(x)$, the origin of the closed-loop system is  asymptotically stable.

        \item In addition, if $b_1\to 0$ as $x\to0$, then the control laws $u=k_{\rm l}^*(x)$ and $u=k_{\rm m}^*(x)$ are also continuous at the origin $x=0$.
    \end{enumerate}
\end{theorem}

\begin{proof}
    See Appendix \ref{appexdix:B}.
\end{proof}

The condition in Theorem~3(iv), $b_1\to 0$ as $x\to 0$, ensures continuity of the feedback at the origin. It is not needed for the safety or stability results in (i)–(iii). We show next that a modified feedback removes this requirement while preserving forward invariance and asymptotic stability.

As shown in the proof of Theorem~\ref{thm:3}, there exists a sufficiently small neighborhood $\mathcal{O}$ of the origin $x=0$ in which the implication~\eqref{eq:14} holds. Consequently, under the control law $u = k_{\rm l}^*(x)$ or $u = k_{\rm m}^*(x)$, the CLF and CBF inequalities \eqref{eq:clf}-\eqref{eq:cbf} are satisfied for all $x \in \mathcal{O} \setminus \{0\}$. Moreover, since $0 \in \operatorname{Int}(\mathcal{C})$, one has $a_1(0) < 0$. By continuity of $a_1(x)$ and $b_1(x)$, there exists a neighborhood $\mathcal{O}$ such that $a_1(x) < 0$ holds for all $x \in \mathcal{O}$. In this neighborhood, the SCP of the CLF ensures that the stabilizing control input can be chosen sufficiently small so that the inequality $a_1(x) + b_1(x)u < 0$ holds for all $x \in \mathcal{O}$. This observation implies that, locally around the origin, the CBF constraint does not become active and does not interfere with stabilization, and thus, motivates the following modification of the control law:
\begin{equation}\label{eq:266}
k_{\rm l}^\sharp(x):=(1-\mu_c(x))k_{\rm l}^*(x)+ \mu_c(x)\phi_0
\end{equation}
where $\mu_c(x):=\frac{1}{1+c|x|^2}$ with $c>0$. Since $\mu_c(x)\to 1$ as $|x|\to 0$ and $\mu_c(x)\to 0$ as $|x|\to\infty$, choosing $c$ sufficiently large ensures that the region where $\mu_c(x)\ge \bar\mu$ is contained in $\mathcal O$. The specific form of $\mu_c$ is not essential; any continuous weighting function with values in $[0,1]$ having the same qualitative properties may be used instead.

\begin{corollary}
Consider the system~\eqref{eq:nls} and assume that there exist a CLF $V$ satisfying the SCP and a CBF $h$, and that $0 \in \operatorname{Int}(C)$. Then, there exists $c_0>0$ such that for all $c>c_0$, the control law $u=k_{\rm l}^\sharp(x)$ is well defined and continuous on $\mathbb{R}^n$. Under the control law $u=k_{\rm l}^\sharp(x)$, the set $\operatorname{Int}(C)$ is forward invariant, and the origin of the closed-loop system is asymptotically stable.
\end{corollary}
\begin{proof}
Well-definedness and continuity of $k_{\rm l}^\sharp(x)$ on $\mathbb{R}^n \setminus \{0\}$ follow directly from Theorem~\ref{thm:3} and the fact that $\mu_c$ is smooth and both $k_{\rm l}^*$ and $\phi_0$ are continuous on $\mathbb{R}^n \setminus \{0\}$. To see continuity at the origin, note that $k_{\rm l}^\sharp(0) = 0$. Moreover, by the SCP of the CLF and the standard properties of the chosen universal formula, one has $\phi_0(x) \to 0$ as $|x| \to 0$. Therefore, 
\begin{equation}
    \lim\limits_{|x| \to 0} k_{\rm l}^\sharp(x)=\lim\limits_{|x| \to 0} \phi_0(x)=0= k_{\rm l}^\sharp(0),
\end{equation}
which proves continuity at the origin. 
(If desired, this limit can be justified by additionally noting that $b_0 \to 0$ and that $\phi_1$ is locally bounded on a neighborhood of the origin, which holds under the standing assumptions in Theorem~\ref{thm:3}.)

We next argue that $k_{\rm l}^\sharp(x)$ preserves safety and asymptotic stability. Since $0 \in \operatorname{Int}(\mathcal{C})$, one has $a_1(0) < 0$. By continuity of $a_1$, there exists a neighborhood $\mathcal{O}$ of the origin such that $a_1(x) < 0$ for all $x \in \mathcal{O}$. In this neighborhood, the SCP of the CLF ensures that the stabilizing input can be chosen sufficiently small so that the CBF inequality remains satisfied, i.e., $a_1(x) + b_1(x)\phi_0(x) < 0, \forall x \in \mathcal{O}$. Choose $c>0$ sufficiently large so that the region $\{x : \mu_c(x) \ge \bar{\mu}\} \subset \mathcal{O}$ for some fixed $\bar{\mu} \in (0,1)$. Then $k_{\rm l}^\sharp$ coincides with $k_{\rm l}^*$ outside $\mathcal{O}$ up to a vanishing weight, and the CBF inequality continues to hold everywhere.  The CLF inequality is handled analogously, using the fact that both $u=k_{\rm l}^*(x)$ and $u=\phi_0(x)$ satisfy it and that the CLF constraint is affine in $u$. Hence, $k_{\rm l}^\sharp$ preserves the safe stabilization property while being continuous everywhere.
\end{proof}

A similar modification can also be applied to $k_{\rm m}^*$, i.e.,
\begin{equation}
k_{\rm m}^\sharp(x):=(1-\mu_c(x))k_{\rm m}^*(x)+ \mu_c(x)\phi_0
\end{equation}
Since the construction follows the same principles and differs only in the choice of the universal formula components, the details are omitted for brevity.

\begin{remark}\rm
We stabilize equilibria in the strict interior of the safe set. The strict barrier inequality produces controls that make the boundary repellent, preventing boundary equilibria and behaviors like ``parameter projection'' in adaptive control or prescribed-time safety \cite{abel2023prescribed_time}.
\end{remark}

\section{Simulation Results}\label{sec:simulation}

In this section, we use an example to illustrate how the proposed framework facilitates the synthesis of a safe stabilizing controller in the presence of nonlinear state constraints.

Consider the system
\begin{subequations}\label{eq:example}
    \begin{eqnarray}
        \dot{x}_1 &=& x_1+\sin x_1 + x_2, \\
        \dot{x}_2 &=& x_1^3 + (1+x_1^2)u.
    \end{eqnarray}
\end{subequations}
The system \eqref{eq:example} is open-loop unstable. Let
\begin{equation}
    V(x) := \frac{1}{2}x_1^2 + \frac{1}{2}\bigl(x_2 + (k_1+1)x_1 + \sin x_1\bigr)^2.
\end{equation}
It can be verified that $V$ is a CLF satisfying the SCP with associated positive definite function
\begin{equation}
    \alpha(x):=\frac{1}{2}\left(k_1x_1^2+k_2\bigl(x_2+(1+k_1)x_1+\sin x_1\bigr)^2\right),
\end{equation}
where $k_1,k_2>0$. We aim to stabilize the origin of \eqref{eq:example} while simultaneously enforcing the nonlinear safety constraint
\begin{equation*}
    \mathcal{C}:=\left\{x\in\mathbb{R}^2:\; h(x):=x_2+q(x_1-d_1)^2+d_2>0\right\},
\end{equation*}
where $d_1,d_2,$ and $q$ are constants. One can verify that $h$ is a CBF for system \eqref{eq:example} with $\alpha_h(s):=s$.

Figure \ref{fig:path2} presents the simulation results for $k_1=1$, $k_2=2$, $q=8$, $d_1=0.5$, and $d_2=1.1$. In the simulation, \eqref{eq:lambda} is used with the parameter choice $c=10^5$ in \eqref{eq:266}. For this example, the coefficients $b_0$ and $b_1$ may take opposite signs, and hence the compatibility condition is not satisfied globally.
As expected, the Sontag-type controller $\phi_{\rm S}(a_0,b_0)$ (green dashed curve) fails to preserve safety and drives the state outside the admissible set. In contrast, both proposed controllers, $u=k_{\rm l}^\sharp(x)$ (blue solid curve) and $u=k_{\rm m}^\sharp(x)$ (red dash-dot curve), constructed from $(\phi_{\rm S}(a_0,b_0),\phi_{\rm S}(a_1,b_1))$, successfully enforce the safety constraint while steering the state toward the origin. To visualize whether the compatibility condition holds, we introduce the indicator $\text{mode}$, where $\text{mode}=1$ means that the compatibility condition is satisfied and $\text{mode}=0$ means that it is not. 
As shown in Fig.~\ref{fig:traj2}, the compatibility condition is satisfied at the beginning of the simulation, then becomes violated over an intermediate time interval, and is eventually restored again. During the interval where $\text{mode}=0$, the closed-loop trajectories under $u=k_{\rm l}^\sharp(x)$ and $u=k_{\rm m}^\sharp(x)$ move toward the boundary of the safe set without crossing it. Once the compatibility condition is recovered, both controllers ensure safety and asymptotic stabilization simultaneously.
  
\begin{figure}[t]
    \centering
    \includegraphics[width=0.6\linewidth]{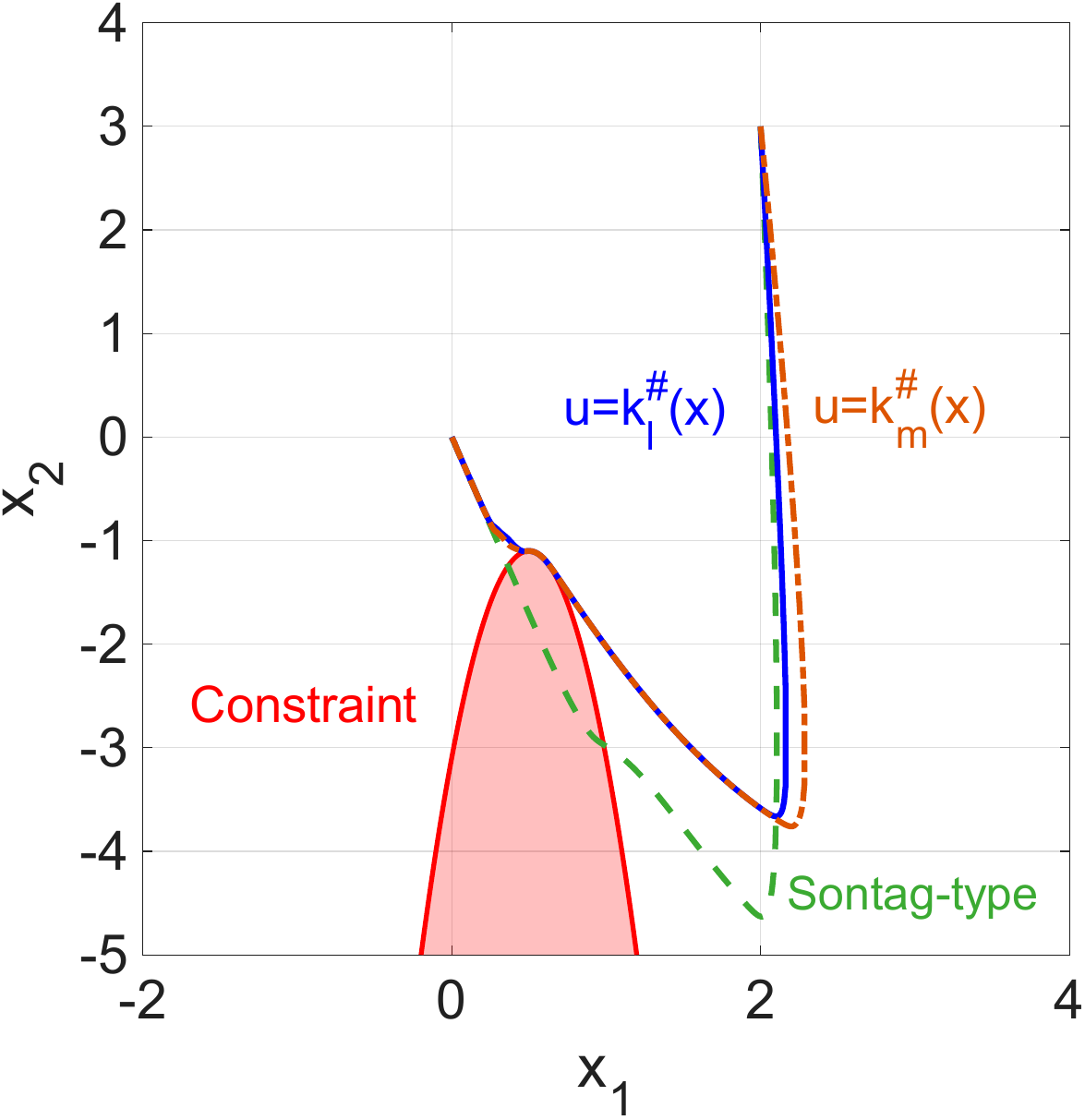}
    \caption{\unboldmath Trajectories of the system under the following inputs:
$u = k_{\rm l}^\sharp(x)$ (blue solid), $u = k_{\rm m}^\sharp(x)$ (red dash-dot), and Sontag-type controller $u = \phi_{\rm S}(a_0,b_0)$ (green dashed), against the constraint (red solid).}
    \label{fig:path2}
\end{figure}

\begin{figure}[t]
    \centering
    \includegraphics[width=0.8\linewidth]{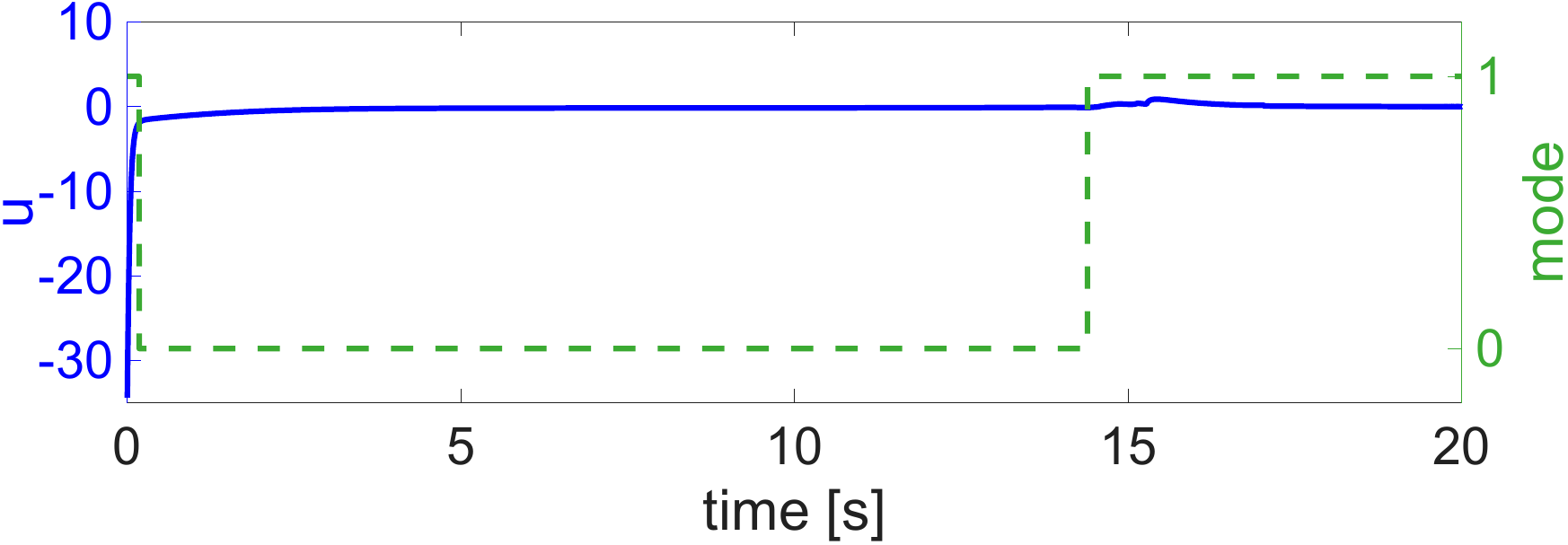}
    \caption{The control action and the mode.}
    \label{fig:traj2}
\end{figure}

\section{Conclusion}\label{sec:conclusion}

This paper shows that, for single-input control-affine nonlinear systems, safe stabilization is governed by an explicit compatibility relation between the CLF and CBF inequalities. When this condition holds, it yields families of closed-form 
feedback laws, parametrized by a free nondecreasing function, 
based on Lie-derivative data and standard universal stabilizers. When it fails, the proposed modification preserves forward invariance while driving the state toward the safe-set boundary, making the safety-stability tradeoff explicit.

\appendices
\section{Proof of Theorem~\ref{thm:2}}\label{appexdix:A}

We first prove that with $u = k_{\rm l}(x)$ and $u = k_{\rm m}(x)$, \eqref{eq:clf}-\eqref{eq:cbf} hold simultaneously (i.e., safe stabilization). Then, we establish the continuity of $k_{\rm l}$ and $k_{\rm m}$.
      \paragraph*{1) Safe Stabilization} We have the following four cases: 
      
      \textbf{Case 1.} ($b_0, b_1\ge 0, |b|\ne 0$) If $b_i > 0$ for $i \in \{0,1\}$, then $\phi_i < -\frac{a_i}{b_i}$, and hence
        \begin{equation}
        b_i u = b_i \min\{\phi_0,\phi_1\} \le b_i \phi_i < -a_i,
        \end{equation}
        which implies that for $i \in \{0,1\}$,
        \begin{equation}
        a_i + b_i u < 0.
        \end{equation}
        If $b_i > 0$ and $b_j = 0$ for $i, j \in \{0,1\}$, then $\phi_i < 0$ and $\phi_j = 0$, and thus $u=\min\{\phi_0, \phi_1\} = \phi_i$. This implies that
        \begin{equation}
        a_i + b_i u < 0
        \quad \text{and} \quad
        a_j + b_j u = a_j < 0.
        \end{equation}      

      \textbf{Case 2.} ($b_0, b_1\le 0, |b|\ne 0$) If $b_i < 0$ for $i \in \{0,1\}$, then $\phi_i > -\frac{a_i}{b_i}$, and hence
        \begin{equation}
        b_i u = b_i \max\{\phi_0,\phi_1\} \le  b_i \phi_i < -a_i,
        \end{equation}
        which implies that for $i \in \{0,1\}$,
        \begin{equation}
        a_i + b_i u < 0.
        \end{equation}
        If $b_i < 0$ and $b_j = 0$ for $i, j \in \{0,1\}$, then $\phi_i > 0$ and $\phi_j = 0$, and thus $u=\max\{\phi_0, \phi_1\} = \phi_i$. This implies that
        \begin{equation}
        a_i + b_i u < 0
        \quad \text{and} \quad
        a_j + b_j u = a_j < 0.
        \end{equation}

      \textbf{Case 3.} ($|b|= 0$)     
      It follows from the definitions of the CLF and CBF in \eqref{eq:defCLF} and \eqref{eq:defCBF} that $b_0 = b_1 = 0$ implies $a_0 < 0$ and $a_1 < 0$. Hence, in this case \eqref{eq:clf}-\eqref{eq:cbf} hold trivially.

      \textbf{Case 4.} ($b_0\cdot b_1<0$) Assume that $b_i(x) > 0$ and $b_j(x) < 0$ for $i, j \in \{0,1\}$.  In this case, \eqref{eq:14} holds. It follows that if
        \begin{equation}
            u\in\left(-\frac{a_j}{b_j}, -\frac{a_i}{b_i}\right),
        \end{equation}
    then \eqref{eq:clf}-\eqref{eq:cbf} hold simultaneously. This is because $u<-\frac{a_i}{b_i}$ implies that
    \begin{equation}
        a_i+b_iu<a_i+b_i\left(-\frac{a_i}{b_i}\right)=0,
    \end{equation}
    and $u>-\frac{a_j}{b_j}$ implies that
    \begin{equation}
        a_j+b_ju<a_j+b_j\left(-\frac{a_j}{b_j}\right)=0.
    \end{equation}
    Hence, it is sufficient to show that when $b_0\cdot b_1<0$, both $k_{\rm l}$ and $k_{\rm m}$ lie in the interval $\left(-\frac{a_j}{b_j}, -\frac{a_i}{b_i}\right)$.

    When $b_i(x) > 0$ and $b_j(x) < 0$, we have $\phi_i < 0$, $\phi_i < -\frac{a_i}{b_i}$, $\phi_j > 0$, and $\phi_j > -\frac{a_j}{b_j}$, as shown in Fig. \ref{fig:kl}. Thus, we have 
    \begin{equation}
        \max\left\{\phi_i,-\frac{a_j}{b_j} \right\}\in \left[ -\frac{a_j}{b_j}, -\frac{a_i}{b_i}\right),
    \end{equation}
    and
    \begin{equation}
        \min\left\{\phi_j,-\frac{a_i}{b_i} \right\}\in \left( -\frac{a_j}{b_j}, -\frac{a_i}{b_i}\right].
    \end{equation}
    Note that $\lambda(\cdot)\in(0,1)$ and that $k_{\rm l}$ is the convex combination of $\max\big\{\phi_i,-\frac{a_j}{b_j} \big\}$ and $\min\big\{\phi_j,-\frac{a_i}{b_i} \big\}$. Hence, $k_{\rm l}(x)\in\big(-\frac{a_j}{b_j}, -\frac{a_i}{b_i}\big)$ when $b_0(x)\cdot b_1(x)<0$, which implies that \eqref{eq:clf}-\eqref{eq:cbf} hold simultaneously with $u=k_{\rm l}(x)$.

    Similarly, it follows from~\eqref{eq:km} that
\begin{equation}
k_{\rm m}(x) \le \max \left\{ \phi_i,\frac{1}{2}\left(-\frac{a_i}{b_i}-\frac{a_j}{b_j}\right)\right\} < -\frac{a_i}{b_i},
\end{equation}
and
\begin{equation}
k_{\rm m}(x) \ge \phi_j > -\frac{a_j}{b_j},
\end{equation}
which shows that, when $b_0(x)\cdot b_1(x) < 0$, \eqref{eq:clf}-\eqref{eq:cbf} hold simultaneously with $u = k_{\rm m}(x)$.
Thus, we conclude that \eqref{eq:clf}-\eqref{eq:cbf} hold simultaneously if $b_0\cdot b_1<0$.
    
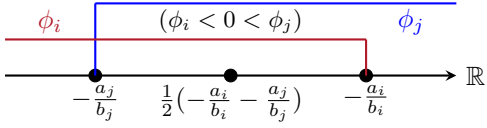
\begin{figure}[t]
\centering
\begin{tikzpicture}[>=stealth, scale=1.2]
\draw[thick,->] (-1,0) -- (4,0) node[right] {$\mathbb{R}$};
\coordinate (A) at (0,0);
\coordinate (B) at (3,0);
\coordinate (C) at (1.5,0);
\draw[fill] (A) circle (2pt) node[below] {$-\frac{a_j}{b_j}$};
\draw[fill] (B) circle (2pt) node[below] {$-\frac{a_i}{b_i}$};
\draw[fill] (C) circle (2pt) node[below] {$\frac{1}{2}(-\frac{a_i}{b_i}-\frac{a_j}{b_j})$};

\draw[thick,blue] (0,0) -- (0,0.8);
\draw[thick,blue] (0,0.8) -- (4,0.8);
\draw[thick,myred] (3,0) -- (3,0.4);
\draw[thick,myred] (3,0.4) -- (-1,0.4);

\node at (1.5,0.6) {\small$(\phi_i<0<\phi_j)$};
\node at (-0.5,0.6) {\color{myred}$\phi_i$};
\node at (3.5,0.6) {\color{blue}$\phi_j$};
\end{tikzpicture}
\caption{Illustration of \unboldmath$\phi_i$ and $\phi_j$ when $b_i > 0$ and $b_j < 0$ and the necessary condition \eqref{eq:14} is satisfied.}
\label{fig:kl}
\end{figure}

      \paragraph*{2) Continuity of $k_{\rm l}$ and $k_{\rm m}$} Next, we show that the two feedback laws $k_{\rm l}$ and $k_{\rm m}$ are continuous on $\mathbb{R}^n\backslash\{0\}$. Since $\phi_0$ and $\phi_1$ are continuous everywhere, $k_{\rm l}$ and $k_{\rm m}$ are also continuous in the ``interiors", i.e., when $b_0\cdot b_1>0$ or $b_0\cdot b_1<0$, as shown in Fig. \ref{fig:b}. The only possible points of discontinuity of $k_{\rm l}$ and $k_{\rm m}$ are those points on the switching surfaces, i.e., $b_0(x)=0$ or $b_1(x)=0$. So, we need to verify the continuity of the functions $k_{\rm l}$ and $k_{\rm m}$ at the points on the $b_0$ and $b_1$ axes. 
      
      We first verify the continuity at the ``origin" ($|b|=0$).
      In both \textbf{Case 1} ($b_0, b_1\ge 0, |b|\ne 0$) and \textbf{Case 2} ($b_0, b_1\le 0, |b|\ne 0$), if $b_0\to 0$ and $b_1\to 0$, then it follows from the continuity of $\phi_0$ and $\phi_1$ that $\phi_0\to 0$ and $\phi_1\to 0$, and thus, $\min\{\phi_0,\phi_1\}\to 0$ and $\max\{\phi_0,\phi_1\}\to 0$, which coincides with the function value $k_{\rm l}(x)=k_{\rm m}(x)=0$ when $|b|=0$. In \textbf{Case 4} ($b_i>0$ and $b_j<0$), if $b_i\to 0_+$ and $b_j\to 0_-$, then we have $-\frac{a_i}{b_i}\to +\infty$ and $-\frac{a_j}{b_j}\to -\infty$, which yields
      \begin{equation*}
      	\max\left\{\phi_i,-\frac{a_j}{b_j} \right\}\to \phi_i\quad  \text{and} \quad  \min\left\{\phi_j,-\frac{a_i}{b_i} \right\} \to \phi_j.
      \end{equation*}
      Thus, in this case $k_{\rm l}(x)=(1-\lambda)\phi_i+\lambda\phi_j\to 0$ as $b_i\to 0_+$ and $b_j\to 0_-$. Also, note that 
      \begin{equation}
      	\min\{0,\max\{0,s\}\}= 0,\quad\forall s\in\mathbb{R}.
      \end{equation}
	  Thus, $k_{\rm m}(x)\to 0$ as $b_i\to 0_+$ and $b_j\to 0_-$. We conclude that $k_{\rm l}$ and $k_{\rm m}$ are continuous on $\{x\in\mathbb{R}^n:|b|=0\}$. 

	Next, we verify the continuity of $k_{\rm l}$ and $k_{\rm m}$ on the set
    \begin{equation*}
        \{b_0=0,b_1\ne 0\}\cup\{b_0\ne 0,b_1=0\}.
    \end{equation*}
    In \textbf{Case~1} ($b_0, b_1 \ge 0$, $|b|\ne 0$), assume that $b_i > 0$ and $b_j = 0$. Then $\phi_i < 0$ and $\phi_j = 0$, and hence $\min\{\phi_0,\phi_1\} = \phi_i$.
    In \textbf{Case~2} ($b_0, b_1 \le 0$, $|b|\ne 0$), assume that $b_i = 0$ and $b_j < 0$. Then $\phi_i = 0$ and $\phi_j > 0$, and hence $\max\{\phi_0,\phi_1\} = \phi_j$, as illustrated in Fig.~\ref{fig:b}. To verify continuity at points on $\{b_0 = 0, b_1 \ne 0\} \cup \{b_0 \ne 0, b_1 = 0\}$, we consider \textbf{Case~4} ($b_i > 0$ and $b_j < 0$) and compute the limits of $k_{\rm l}(x)$ and $k_{\rm m}(x)$ in the following two situations\footnote{The remaining situation $b_i \to 0_+$, $b_j \to 0_-$ has been discussed earlier.}:
    \begin{enumerate}[a)]
    \item ($b_i \to 0_+$ and $b_j < 0$)
    In this case, $-\frac{a_i}{b_i} \to +\infty$ while $-\frac{a_j}{b_j}$ remains finite.
    Hence $\lambda(\cdot) \to 1$, and therefore
    \begin{equation}
        k_{\rm l}(x)\to \min\{\phi_j,+\infty\}=\phi_j
    \end{equation}
    Likewise, one gets
    \begin{equation}
		k_{\rm m}(x)\to \min\left\{\phi_j,\max \left\{ \phi_i,+\infty\right\}\right\}=\phi_j,
	\end{equation}
    \item ($b_i > 0$ and $b_j \to 0_-$)  
In this case, $-\frac{a_i}{b_i}$ is finite while $-\frac{a_j}{b_j} \to -\infty$.  
Hence $\lambda(\cdot) \to 0$, and thus
\begin{equation}
    k_{\rm l}(x) \to \max\{\phi_i, -\infty\} = \phi_i.
\end{equation}
Similarly, one gets
\begin{equation}
    k_{\rm m}(x) \to \min\!\left\{\phi_j,\; \max\{\phi_i, -\infty\}\right\} = \phi_i,
\end{equation}
because in this case, $\phi_i<0$ and $\phi_j>0$.
    \end{enumerate}
    In both cases, the limiting values of $k_{\rm l}$ and $k_{\rm m}$ agree with their function values on
$\{b_0=0, b_1\ne 0\} \cup \{b_0\ne 0, b_1=0\}$.
Therefore, we conclude that $k_{\rm l}$ and $k_{\rm m}$ are continuous on $\mathbb{R}^n\backslash\{0\}$, which completes the proof. \qed
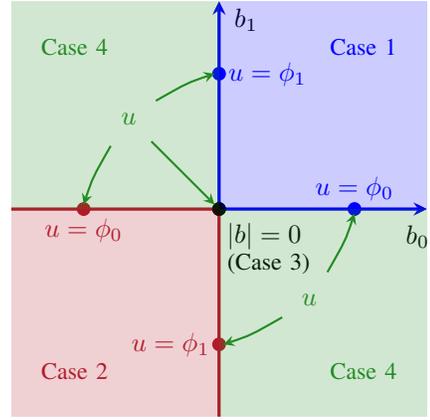
\begin{figure}[t]
\centering
\begin{tikzpicture}[>=stealth, scale=1.2]
\draw[thick,->] (-2.3,0) -- (2.3,0) node at (2.2,-0.3) {$b_0$};
\draw[thick,->] (0,-2.3) -- (0,2.3) node at (0.3, 2.1) {$b_1$};
\draw[blue, very thick,->] (0,0) -- (2.3,0);
\draw[blue, very thick,->] (0,0) -- (0,2.3);

\draw[myred, very thick,-] (0,0) -- (-2.3,0);
\draw[myred, very thick,-] (0,0) -- (0,-2.3);

\coordinate (A) at (0,0);
\draw[fill] (A) circle (2pt) node at (0.5,-0.3) {$|b|=0$};
\draw[fill] (A) circle (2pt) node at (0.55,-0.6) {\small (Case 3)};

\coordinate (B) at (1.5,0);
\draw[fill,blue] (B) circle (2pt) node[above] {$u=\phi_0$};

\coordinate (C) at (0,1.5);
\draw[fill,blue] (C) circle (2pt) node[right] {$u=\phi_1$};

\coordinate (D) at (-1.5,0);
\draw[fill,myred] (D) circle (2pt) node[below] {$u=\phi_0$};

\coordinate (E) at (0,-1.5);
\draw[fill,myred] (E) circle (2pt) node[left] {$u=\phi_1$};

\node at (1.6,1.8) {\color{blue}\small Case 1};
\node at (-1.6,-1.8) {\color{myred}\small Case 2};
\node at (-1.6,1.8) {\color{mygreen}\small Case 4};
\node at (1.6,-1.8) {\color{mygreen}\small Case 4};

\fill[blue, fill opacity=0.2] (0,0) rectangle (2.3,2.3);
\fill[myred, fill opacity=0.2] (0,0) rectangle (-2.3,-2.3);
\fill[mygreen, fill opacity=0.2] (0,0) rectangle (2.3,-2.3);
\fill[mygreen, fill opacity=0.2] (0,0) rectangle (-2.3,+2.3);

\draw[mygreen, ->, thick] (-0.7,1.2) .. controls (-0.35,1.4) .. (-0.05,1.5);
\draw[mygreen, ->, thick] (-1.2,0.7) .. controls (-1.4,0.35) .. (-1.5,0.05);
\draw[mygreen, ->, thick] (-0.75,0.75) .. controls (-0.35,0.35) .. (-0.05,0.05);
\node at (-1,1) {\color{mygreen}$u$};

\draw[mygreen, ->, thick] (0.7,-1.2) .. controls (0.35,-1.4) .. (0.05,-1.5);
\draw[mygreen, ->, thick] (1.2,-0.7) .. controls (1.4,-0.35) .. (1.5,-0.05);
\node at (1,-1) {\color{mygreen}$u$};
\end{tikzpicture}
\caption{Illustration of the continuity of \unboldmath$k_{\rm l}(x)$ and $k_{\rm m}(x)$.}
\label{fig:b}
\end{figure}

\section{Proof of Theorem~\ref{thm:3}}\label{appexdix:B}

    \paragraph*{1) Well-Definedness and Continuity} Both $k_{\rm l}$ and $k_{\rm m}$ are well defined and continuous following Theorem \ref{thm:2}. We only need to focus on the case where $b_0 b_1 < 0$  and \eqref{eq:14} does not hold, as shown in Fig. \ref{fig:3}. In this case, $u=-\frac{a_1}{b_1}$, and we first prove that it is impossible for $b_1$ to approach zero.

    Suppose that $b_1 > 0$ and $b_1 \to 0_+$ (i.e., $i = 1$ and $j = 0$). Then it is necessary that $-\frac{a_i}{b_i}=-\frac{a_1}{b_1}\to +\infty$, and hence it must also be the case that $-\frac{a_j}{b_j}<-\frac{a_i}{b_i}\to +\infty$, where \eqref{eq:14} is satisfied. On the other hand, suppose that $b_1 < 0$ and $b_1 \to 0_-$ (i.e., $i = 0$ and $j = 1$). Then it is necessary that $-\frac{a_j}{b_j}=-\frac{a_1}{b_1}\to -\infty$, and hence it must also be the case that $-\infty \leftarrow-\frac{a_j}{b_j}<-\frac{a_i}{b_i}$, where \eqref{eq:14} is also satisfied. Thus, when $b_0 b_1 < 0$  and \eqref{eq:14} does not hold, it is impossible for $b_1$ to approach zero, and hence $k_{\rm l}^*$ and $k_{\rm m}^*$ are well defined.

    Furthermore, the control law $u=-\frac{a_1}{b_1}$ is continuous when $b_1$ remains bounded away from zero. Also, from the analysis above, if $b_1\to 0$, then it is necessary that \eqref{eq:14} is satisfied (i.e., \textbf{Case 4} in the proof of Theorem \ref{thm:2}). So, we only need to verify the continuity of the control laws $k_{\rm l}^*$ and $k_{\rm m}^*$ when $b_0b_1<0$ and $-\frac{a_j}{b_j}=-\frac{a_i}{b_i}$. Let us consider \eqref{eq:kl} and $\eqref{eq:km}$ and suppose that $-\frac{a_j}{b_j}-\big(-\frac{a_i}{b_i}\big)\to 0_-$. It follows that
    \begin{equation}
    k_{\rm l}(x)=(1-\lambda)\cdot\left(-\frac{a_j}{b_j} \right) + \lambda\cdot\left(-\frac{a_i}{b_i} \right)\to -\dfrac{a_1}{b_1}
    \end{equation}
    and that  
    \begin{equation}
    k_{\rm m}(x)\to\min\left\{\phi_j,\max \left\{ \phi_i,-\frac{a_1}{b_1}\right\}\right\}=-\dfrac{a_1}{b_1}.
    \end{equation}
    Hence, both the control laws $u=k_{\rm l}^*(x)$ and $u=k_{\rm m}^*(x)$ are continuous on $\mathbb{R}^n\backslash\{0\}$.

    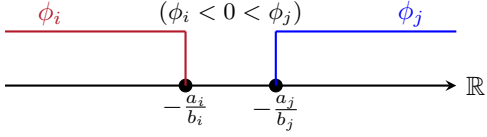
\begin{figure}[t]
\centering
\begin{tikzpicture}[>=stealth, scale=1.2]
\draw[thick,->] (-1,0) -- (4,0) node[right] {$\mathbb{R}$};
\coordinate (B) at (1,0);
\coordinate (A) at (2,0);
\coordinate (C) at (1.5,0);
\draw[fill] (A) circle (2pt) node[below] {$-\frac{a_j}{b_j}$};
\draw[fill] (B) circle (2pt) node[below] {$-\frac{a_i}{b_i}$};

\draw[thick,myred] (1,0) -- (1,0.6);
\draw[thick,myred] (1,0.6) -- (-1,0.6);
\draw[thick,blue] (2,0) -- (2,0.6);
\draw[thick,blue] (2,0.6) -- (4,0.6);

\node at (1.5,0.8) {\small$(\phi_i<0<\phi_j)$};
\node at (-0.5,0.8) {\color{myred}$\phi_i$};
\node at (3.5,0.8) {\color{blue}$\phi_j$};
\end{tikzpicture}
\caption{Illustration of \unboldmath$\phi_i$ and $\phi_j$ when $b_i > 0$ and $b_j < 0$ and the necessary condition \eqref{eq:14} is not satisfied.}
\label{fig:3}
\end{figure}

    \paragraph*{2) Safety} If $b_0b_1\ge 0$, or if $b_0b_1<0$ while condition \eqref{eq:14} is satisfied, it follows from Theorem \ref{thm:2} that $F_1:=a_1(x)+b_1(x)u<0$ with $u=k_{\rm l}^*(x)$ or $u=k_{\rm m}^*(x)$, which implies that
    \begin{equation}\label{eq:43}
        \dot{h}=L_fh(x)+L_gh(x)u>-\alpha_h(h(x)).
    \end{equation}
    If $b_0b_1<0$ and the condition \eqref{eq:14} does not hold, then $u=-\frac{a_1}{b_1}$ and $F_1:=a_1(x)+b_1(x)u=0$, which implies that 
    \begin{equation}\label{eq:44}
        \dot{h}=L_fh(x)+L_gh(x)u=-\alpha_h(h(x)).
    \end{equation}
    Along trajectories, \eqref{eq:43}-\eqref{eq:44} guarantees that the set $\operatorname{Int}(\mathcal{C})$ is forward invariant \cite{ames2017control}. In particular, \eqref{eq:44} indicates that the trajectory of the system moves toward the boundary $\partial \mathcal{C}$ of the safe set without crossing it.
    
    \paragraph*{3) Stability} To show asymptotic stability of the origin $0\in\operatorname{Int}(\mathcal{C})$, we need to prove that, in a neighborhood of $x=0$, if $b_0 b_1 < 0$ then condition~\eqref{eq:14} must hold. Once this is shown, it follows from Theorem~\ref{thm:2} that \eqref{eq:clf} is satisfied, and thus asymptotic stability of the origin is guaranteed.

    In a neighborhood of the origin, by the SCP, there exists a continuous control law $u_c$ such that $a_0+b_0u_c<0$. Since $u_c\to 0$ as $x\to 0$, we have $\frac{|a_0|}{|b_0|}\to 0$ as $x\to 0$. Moreover, note that $a_1(0)=-\alpha_h(h(0))<0$.

    Now, in a neighborhood of the origin, if $b_0b_1<0$, there are two possible cases:
\begin{enumerate}[(C1)]
\item ($b_1>0$) In this case, $i=1$ and $j=0$. Then $-\frac{a_i}{b_i}>-\frac{a_j}{b_j}\to 0$. 
\item ($b_1<0$) In this case, $i=0$ and $j=1$. Then $-\frac{a_j}{b_j}<-\frac{a_i}{b_i}\to 0$. 
\end{enumerate}
In either case, if $b_0b_1<0$, condition~\eqref{eq:14} necessarily holds in a neighborhood of the origin. Consequently, asymptotic stability of the origin follows.

\paragraph*{4) Continuity at $x=0$} As $x \to 0$, we have $|b| \to 0$ and hence $u \to 0$, which establishes continuity at the origin. \qed

\bibliographystyle{ieeetr} 
\bibliography{mybibfile}   
\end{document}